\documentclass[9pt,journal]{IEEEtran}

\usepackage{mathrsfs}
\usepackage{amsfonts}
\usepackage{amsmath,amssymb}
\usepackage{graphicx}
\usepackage{float}
\usepackage[dvips]{epsfig}
\usepackage[dvips]{color}
\usepackage{subfigure}
\usepackage{caption}

\def\dref#1{(\ref{#1})}

\newtheorem{assumption}{Assumption}
\newtheorem{lemma}{Lemma}
\newtheorem{theorem}{Theorem}
\newtheorem{corollary}{Corollary}
\newtheorem{remark}{Remark}

\newtheorem{proof}{Proof}
\newtheorem{algorithm}{Algorithm}

\begin{document}

%
\title{Fully Distributed Event-Triggered Protocols for Linear Multi-Agent Networks}

\author{Bin Cheng and Zhongkui Li,~\IEEEmembership{Member,~IEEE}
\thanks{This work was supported in part by the National Natural Science Foundation of China under grants 61473005, 11332001, U1713223, and by Beijing Nova Program under grant 2018047. Corresponding author: Zhongkui Li.}
\thanks{B. Cheng and Z. Li are with the State Key Laboratory for Turbulence and Complex Systems, Department of Mechanics and Engineering Science, College of Engineering, Peking University, Beijing 100871, China. E-mail: {\tt \{bincheng,zhongkli\}@pku.edu.cn}
}}

\IEEEtitleabstractindextext{%
\begin{abstract}
This paper considers the distributed event-triggered consensus problem for general linear multi-agent networks.
Both the leaderless and leader-follower consensus problems are considered.
Based on the local sampled state or local output information, distributed adaptive event-triggered protocols are designed,
which can ensure that consensus of the agents is achieved and the Zeno behavior is excluded by showing that the interval between any two triggering events is lower bounded by a strictly positive value.
Compared to the previous related works, our main contribution is that the proposed adaptive event-based protocols are fully distributed and scalable, which do not rely on any global information of the network graph and are independent of the network's scale.
In these event-based protocols, continuous communications are not required for either control laws updating or triggering functions monitoring.
\end{abstract}

}

\maketitle

\IEEEdisplaynontitleabstractindextext

\IEEEpeerreviewmaketitle

\section{Introduction}
Cooperative control, having broad applications in various areas including flocking, formation control, and distributed sensor networks \cite{olfatisaber2007consensus,WRen2007information,ZLi2015designing,ZLi2014cooperative,lewis2014cooperative}, relies on the information exchange between neighboring agents and over the network.
The information exchanges among agents are conducted over digital networks consisting of various communication links.
Traditional control strategies require communication process being carried out at each time for continuous-time algorithms or at all iterations for discrete-time algorithms.
However, the bandwidth of the communication network and the power source of the agents are inevitably constrained in many practical systems those have become increasingly networked, wireless, and spatially distributed.
In networked control applications, it makes sense to only transmit
information when certain signal in the system is larger than a threshold value \cite{Henningsson2008sporadic}.
The event-triggered control strategy is developed as an important means for avoiding continuous communications.
Event-based control offers some clear advantages with respect to the traditional methods, such as saving the energy and minimizing the number of control actions, when handling practical constraints, but it also introduces new theoretical and practical issues \cite{Heemels2012an}.
Pioneer works addressing event-based implementations of feedback control laws include \cite{Henningsson2008sporadic,Astrom1999comparison,PTabuada2007event,Heemels2008analysis}.

The event-triggered consensus problem has been widely studied in the past decade.
A core task in the event-triggered consensus problem is to design distributed event-based protocols, consisting of the event-based control laws and the triggering functions.
In \cite{DVDimarogonas2012distributed,Garcia2013decentralised,Meng2013event}, event-triggered and self-triggered consensus algorithms are proposed for single-integrator agents over undirected connected communication topologies.
Decentralized event-triggered consensus algorithms are presented in \cite{Seyboth2013event}  for both single- and double-integrator multi-agent systems.
Event-based consensus problem of multi-agent systems with general linear models was studied in \cite{DYang2016Decentralized,HZhang2014observer,Zhu2014event,Guo2014a,Cheng2016Event,Garcia2014decentralized,TCheng2017event}.
In particular, \cite{DYang2016Decentralized,HZhang2014observer,Zhu2014event,Guo2014a,Garcia2014decentralized} presented several state feedback and observer-based output feedback event-triggered consensus protocols for linear multi-agent networks, and
\cite{Cheng2016Event,TCheng2017event} considered the event-triggered leader-follower consensus problem for multi-agent systems in the presence of one leader with fixed and switching topologies.
It is worth noting that as pointed out in \cite{Seyboth2013event,DYang2016Decentralized,Zhu2014event}, certain global information of the network, in terms of nonzero eigenvalues of the Laplacian matrix associated with the communication graph, is generally required in the aforementioned papers to determine some parameters in either the control laws or the triggering functions.
Therefore, the event-based consensus protocols in the aforementioned works are actually not fully distributed.
To the best of our knowledge, how to design fully distributed event-triggered consensus protocols for general linear multi-agent networks is still open and awaits breakthrough.

In this paper, we devote to designing fully distributed and scalable event-triggered consensus protocols for general linear multi-agent networks. Since the event-based protocols are expected to be scalable, whose design is independent of the scale of the network, the simple method
of estimating a lower bound of the nonzero eigenvalues of the Laplacian matrix using the number of agents, as suggested in \cite{Seyboth2013event,DYang2016Decentralized}, is not applicable.
Note that distributed procedures via local interactions among neighboring agents are proposed in \cite{yang2010decentralized,franceschelli2013decentralized} to estimate the eigenvalues of the Laplacian matrix. Nevertheless, when running these distributed estimation procedures simultaneously with consensus protocols, especially for the case with event-triggered communications, the convergence and feasibility of the overall algorithm remain unclear and questionable.
Therefore, we have to come up with novel perspectives to design fully distributed and scalable event-based protocols.

We consider the leaderless and leader-follower consensus problems for undirected graphs and leader-follower graphs, respectively.
For the case of leaderless consensus, we propose a distributed adaptive event-based protocol, based on the sampled state information of neighboring agents.
One distinct feature of the proposed adaptive event-based protocol is that it includes time-varying weights into both the control laws and the triggering functions.
We show that the proposed protocol can guarantee consensus and is robust with respect to bounded external disturbances. We further rule out the Zeno behavior at any finite time by showing that the interval between two arbitrary triggering instants is lower bounded by a strictly positive value.
When the relative state information of neighboring agents is not accessible, we design a distributed event-triggered observer-based adaptive protocol by using only the local output information.
We also extend to consider the leader-follower consensus problem when a leader exists and its information can be received by at least one follower at the initial time instant.
The adaptive event-triggered protocols in this paper rely on the sampled local information of each agent and from its neighbors and do not need continuous communications in either control laws updating or triggering conditions monitoring.
These event-triggered protocols here can be designed and utilized in a fully distributed fashion, i.e., using only the local information of each agent and its neighbors.
Compared to the existing works, e.g., \cite{Seyboth2013event,DYang2016Decentralized,HZhang2014observer,Zhu2014event,Cheng2016Event,Garcia2014decentralized,TCheng2017event}, the main contribution of this paper is that we propose for the first time fully distributed and scalable adaptive event-based protocols, which do not rely on any global information of the network graph and are independent of the network's scale.

The rest of this paper is organized as follows.
The problem statement and motivations are given in Section \ref{s2}.
Fully distributed adaptive event-based consensus protocols are proposed for leaderless graphs in Section \ref{s3}.
Extensions to the case of leader-follower graphs are given in Section \ref{s-leader}.
Numerical simulations are conducted for illustration in Section \ref{s-simulation}.
Finally, Section \ref{s_con} concludes this paper.

\section{Problem Statement and Motivations}\label{s2}

Consider a group of $N$ identical agents with continuous-time general linear dynamics. The dynamics of the $i$-th agent are described by
\begin{equation}\label{sys}
\begin {aligned}
\dot x_i&=Ax_i+Bu_i,\\
y_i&=Cx_i,~i=1,\cdots,N,
\end{aligned}
\end{equation}
where $x_i\in \mathbf{R}^n$, $u_i\in \mathbf{R}^p$, and $y_i\in\mathbf{R}^q$ are, respectively, the state, the control input, and the output of the $i$-th agent.

The communication topology among the $N$ agents is
represented by a directed graph $\mathcal {G}=(\mathcal {V}, \mathcal
{E})$, where $\mathcal {V}=\{v_1,v_2,\cdots,v_N\}$ is the node set
 and $\mathcal {E}\subseteq\mathcal {V}\times\mathcal
{V}$ is the edge set, in which an edge is represented by a pair of distinct nodes. If $(v_i,v_j)\in\mathcal {E}$, node $v_i$
is called a
neighbor of node $v_j$ and node $v_j$ is called an out-neighbor of node $v_i$.
A directed path from
node $v_{i_1}$ to node $v_{i_l}$ is a sequence of adjacent edges of
the form $(v_{i_k}, v_{i_{k+1}})$, $k=1,\cdots,l-1$.
A directed graph contains a directed spanning tree if
there exists a root node that has directed paths to all other nodes.
A graph is said to be undirected, if $(v_i,v_j)\in \mathcal E$ as long as $(v_j,v_i)\in \mathcal E$.
An undirected graph
is connected if there exists a path between every pair
of distinct nodes, otherwise is disconnected.
For a graph $\mathcal G$, its adjacency matrix, denoted by $\mathcal A=[a_{ij}]\in \mathbf R^{N\times N}$, is defined such that $a_{ii}=0$, $a_{ij}=1$ if $(v_j,v_i)\in \mathcal E$ and $a_{ij}=0$ otherwise.
The Laplacian matrix $\mathcal L=[l_{ij}]\in \mathbf R^{N\times N}$ associated with $\mathcal G$ is defined as $l_{ii}=\sum_{j=1}^N{a_{ij}}$ and $l_{ij}=-a_{ij}$, $i\neq j$.
The degree of
agent $v_i$ is defined as $d_i = l_{ii}$.

\begin{lemma}\cite{ZLi2014cooperative}\label{lemma_eigenvalue}
Zero is an eigenvalue of $\mathcal L$ with $\mathbf{1}$ as a right
eigenvector and all nonzero eigenvalues have positive real parts, where ${\bf 1}$ denotes a column vector with all entries equal to 1.
(i) For a directed graph, zero is a simple eigenvalue of $\mathcal L$ if and only if $\mathcal G$ has a
directed spanning tree.
(ii) For an undirected graph, zero is a simple eigenvalue of $\mathcal L$ if and only if $\mathcal G$ is connected. The smallest nonzero eigenvalue $\lambda_2(\mathcal L)$ of $\mathcal L$ satisfies
$\lambda_2(\mathcal L)=\min_{x\neq 0,{\bf 1}^Tx=0}{\frac{x^T\mathcal Lx}{x^Tx}}$.
\end{lemma}

The objective of this paper is to address the event-triggered consensus problem for the agents in \dref{sys}, by ensuring that $\lim_{t\rightarrow \infty}\|x_i-x_j\|=0$, $i,j=1,\cdots,N$, and excluding the Zeno behavior, i.e., there does not exist an infinite number of events within a finite period of time \cite{DVDimarogonas2012distributed,DYang2016Decentralized}.
The crucial task in the event-triggered consensus problem is to design distributed event-based consensus protocols or schemes which consist of the event-based control laws and the triggering functions for the agents.
The control laws rely on the local information, sampled at discrete event time instants.
And the triggering functions determine the event instants, at which time each agent broadcasts its state over the network. Existing event-triggered consensus protocols in, e.g., \cite{DYang2016Decentralized,HZhang2014observer,Cheng2016Event,Garcia2014decentralized,TCheng2017event},
are not truly distributed, requiring the knowledge of global eigenvalue information of the communication graph. This motivates us to remove the limitation in this paper by presenting fully distributed and scalable event-triggered consensus schemes.

\section{Fully Distributed Event-based Protocols For Leaderless Consensus}\label{s3}
In this section, we will design fully distributed event-triggered protocols for leaderless graphs.
The following assumption is needed.
\begin{assumption}\label{ABC}
The pair $(A,B,C)$ in \dref{sys} is stabilizable and detectable, and the graph $\mathcal G$ is undirected and connected.
\end{assumption}

\subsection{State-Based Adaptive Event-Triggered Protocols}\label{s3-1}
In this subsection, we consider the simple case where the relative state information of neighboring agents is available.

Define the state estimate as $\tilde x_i(t)=e^{A(t-t_k^i)}x_i(t_k^i)$, $\forall t\in [t_k^i,t_{k+1}^i)$, where $t_k^i$ denotes the $k$-th event triggering instant of agent $v_i$.
The triggering time instants $t_0^i$, $t_1^i$, $\cdots$, will be determined by the triggering function to be designed later.
For agent $v_i$, we define a measurement error $e_i(t)$ as
\begin{equation}\label{ei}
\begin {aligned}
e_i(t)\triangleq \tilde x_i(t)-x_i(t), i=1,\cdots,N.
\end{aligned}
\end{equation}

Inspired by the adaptive consensus protocols with continuous communications in \cite{ZLi2014cooperative,ZLi2015designing}, we propose a distributed event-based state feedback adaptive control law for each agent as
\begin{equation}\label{pro1}
\begin {aligned}
u_i(t)&=K\sum_{j=1}^N{c_{ij}(t)a_{ij}(\tilde x_i(t)-\tilde x_j(t))},\\
\dot c_{ij}(t)&=
\kappa_{ij}a_{ij}[-\varrho_{ij}c_{ij}(t)+(\tilde x_i(t)-\tilde x_j(t))^T\Gamma(\tilde x_i(t)-\tilde x_j(t))],\\
&\quad\quad\quad\quad~i=1,\cdots,N,
\end{aligned}
\end{equation}
where $c_{ij}(t)$ denotes the time-varying coupling weight for the edge $(v_i,v_j)$ with $c_{ij}(0)=c_{ji}(0)$, $\varrho_{ij}=\varrho_{ji}$ and $\kappa_{ij}=\kappa_{ji}$ are positive constants, and $K\in \mathbf R^{p\times n}$ and $\Gamma \in \mathbf R^{n\times n}$ are the feedback gain matrices.

The triggering function for each agent is given by
\begin{equation}\label{eve}
\begin {aligned}
f_i(t)&=\sum_{j=1}^N(1+\delta c_{ij})a_{ij}e_i^T\Gamma e_i\\
&\quad-\frac{1}{4}\sum_{j=1}^N{a_{ij}(\tilde x_i-\tilde x_j)^T\Gamma(\tilde x_i-\tilde x_j)}-\mu e^{-\nu t},
\end{aligned}
\end{equation}
where $\delta$, $\mu$, and $\nu$ are positive constants.
The event triggering instant is defined as $t_{k+1}^i\triangleq \inf\{t>t_k^i|f_i(t)\geq 0\}$ with $t_0^i=0$.
At the event instant time, agent $v_i$ updates its controller \dref{pro1} using its current state and broadcasts its current state to its neighbors.
Meanwhile, the measurement error $e_i(t)$ is reset to zero.
When the agents receive new states broadcast by any of their neighbors, they will update their controllers immediately.

\begin{remark}
One distinct feature of the adaptive event-triggered protocol in this section is that it includes time-varying weights $c_{ij}(t)$ into both the control law \dref{pro1} and the triggering function \dref{eve}. As a consequence, the triggering function \dref{eve} here is non-quadratic in terms of the measurement error $e_i$, which is different from those in the previous works \cite{Seyboth2013event,DYang2016Decentralized,Zhu2014event}.
The triggering function \dref{eve} consists of a state-dependent term (i.e., the second term on the right-hand side) and a time-dependent term (i.e., the last exponential decay term).  Similarly as in \cite{Nowzari17Event}, we call \dref{eve} a hybrid or mixed triggering function.  Combining the state-dependent and time-dependent terms in \dref{eve}  is expected to be able to rule out the Zeno behavior and meanwhile decrease the event triggering number. 
Note that the event-based protocol, composed of \dref{pro1} and \dref{eve}, is model-based \cite{Garcia2014decentralized}, and relies on the sampled state information of each agent and from its neighbors, rather than agents' real states. Each agent does not need continuously monitor its neighbors' states and therefore neighboring agents do not need continuous communications.
\end{remark}

Denote $\xi=[\xi_1^T,\cdots,\xi_N^T]^T$, where $\xi_i\triangleq x_i-\frac{1}{N}\sum_{j=1}^N{x_j}$.
We can write $\xi$ in a compact form as $\xi=(M\otimes I_n)x$,
where $\otimes$ denotes the Kronecker product, $M=I_N-\frac{1}{N}{\bf{1}\bf{1}}^T$, and $x=[x_1^T,\cdots,x_N^T]^T$.
It is clear that 0 is a simple eigenvalue of $M$ with $\bf 1$ as the corresponding  eigenvector and 1 is the other eigenvalue with multiplicity $N-1$.
It is not difficult to see that $M\mathcal L=\mathcal L=\mathcal LM$.
Then, it follows that $\xi=0$ if and only if $x_1=\cdots=x_N$.
Thus, we can refer to $\xi$ as the consensus error.
Using \dref{pro1} for \dref{sys}, it follows that $\xi$ satisfies the following dynamics:
\begin{equation}\label{xid}
\begin {aligned}
\dot \xi_i&=A\xi_i+BK\sum_{j=1}^N{c_{ij}(t)a_{ij}(\tilde x_i-\tilde x_j)},\\
\dot c_{ij}(t)&=\kappa_{ij}a_{ij}[-\varrho_{ij}c_{ij}(t)+(\tilde x_i-\tilde x_j)^T\Gamma(\tilde x_i-\tilde x_j)].
\end{aligned}
\end{equation}

In what follows, we present an algorithm to construct the event-triggered
adaptive consensus protocol composed of \dref{pro1} and \dref{eve}.

\begin{algorithm}\label{algo1}
Assuming that Assumption \ref{ABC} holds, the event-triggered adaptive consensus protocol consisting of \dref{pro1} and \dref{eve} can be designed according to the following three steps.\\
1) Solve the following algebraic Riccati equation (ARE):
\begin{equation}\label{lyaine}
\begin {aligned}
PA+A^TP-PBB^TP+I=0,
\end{aligned}
\end{equation}
to get a solution $P>0$.\\
2) Choose the feedback matrices $K=-B^TP$ and $\Gamma=PBB^TP$.\\
3) Select $\kappa_{ij}$, $\varrho_{ij}$, $\delta$, $\mu$, and $\nu$ to be any positive constants.
\end{algorithm}

We are now ready to present the main results of this subsection.

\begin{theorem}\label{theorem_1}
Suppose that Assumption \ref{ABC} holds.
Both the consensus error $\xi$ and the coupling gains $c_{ij}$ in \dref{xid} are uniformly ultimately bounded under the event-triggered adaptive protocol constructed by Algorithm \ref{algo1}. If $\varrho_{ij}$ in \dref{pro1} are chosen such that $\varrho_{ij}\kappa_{ij}<1/\lambda_{\max}(P)$, then $\xi$ exponentially converges to a small adjustable bounded set as given in \dref{D2}.
\end{theorem}

\begin{proof}
Consider the Lyapunov function candidate
\begin{equation}\label{lya1}
\begin {aligned}
V_1=\frac{1}{2}\sum_{i=1}^N{\xi_i^TP\xi_i}+\sum_{i=1}^N\sum_{j=1,j\neq i}^N{\frac{(c_{ij}-\alpha)^2}{8\kappa_{ij}}},
\end{aligned}
\end{equation}
where $\alpha$ is a positive constant to be determined later.
Evidently, $V_1$ is positive definite.
The time derivative of $V_1$ along the trajectory of \dref{xid} is given by
\begin{equation}\label{lya1d1}
\begin {aligned}
\dot V_1
&=\sum_{i=1}^N{\xi_i^TP\dot \xi_i}+\sum_{i=1}^N\sum_{j=1,j\neq i}^N{\frac{c_{ij}-\alpha}{4\kappa_{ij}}\dot c_{ij}}\\
&=\sum_{i=1}^N{\xi_i^TPA\xi_i}+
\sum_{i=1}^N{\xi_i^TPBK\sum_{j=1}^Nc_{ij}a_{ij}(\tilde x_i-\tilde x_j)}\\
&\quad+\sum_{i=1}^N\sum_{j=1,j\neq i}^N{\frac{c_{ij}-\alpha}{4\kappa_{ij}}\dot c_{ij}}.
\end{aligned}
\end{equation}
Since $a_{ij}=a_{ji}$ and $c_{ij}(t)=c_{ji}(t)$, it can be easily verified that
\begin{equation}\label{verify}
\begin {aligned}
&\sum_{i=1}^N{\xi_i^TPBK\sum_{j=1}^Nc_{ij}a_{ij}(\tilde x_i-\tilde x_j)}\\
&\qquad=-\frac{1}{2}\sum_{i=1}^N\sum_{j=1}^Nc_{ij}a_{ij}{(\xi_i-\xi_j)^T\Gamma(\tilde x_i-\tilde x_j)}.
\end{aligned}
\end{equation}
By substituting \dref{verify} into \dref{lya1d1}, we have
\begin{equation}\label{lya1d1'}
\begin {aligned}
\dot V_1
=&\sum_{i=1}^N{\xi_i^TPA\xi_i}-
\frac{1}{2}\sum_{i=1}^N\sum_{j=1}^Nc_{ij}a_{ij}{(\tilde x_i-\tilde x_j)^T\Gamma(\tilde x_i-\tilde x_j)}\\
&+\frac{1}{2}\sum_{i=1}^N\sum_{j=1}^Nc_{ij}a_{ij}{(e_i-e_j)^T\Gamma(\tilde x_i-\tilde x_j)}\\&+\sum_{i=1}^N\sum_{j=1,j\neq i}^N{\frac{c_{ij}-\alpha}{4\kappa_{ij}}\dot c_{ij}},
\end{aligned}
\end{equation}
where we have used the facts that $\xi_i-\xi_j=x_i-x_j$ and $e_i=\tilde x_i-x_i$.
In light of the Young's inequality \cite{DBernstein2009Matrix}, it is not difficult to obtain that
\begin{equation}\label{young1}
\begin {aligned}
&\sum_{i=1}^N\sum_{j=1}^Nc_{ij}a_{ij}{(e_i-e_j)^T\Gamma(\tilde x_i-\tilde x_j)}\\
&\qquad\leq \frac{1}{2}\sum_{i=1}^N\sum_{j=1}^Nc_{ij}a_{ij}{(e_i-e_j)^T\Gamma(e_i-e_j)}\\
&\qquad\quad+\frac{1}{2}\sum_{i=1}^N\sum_{j=1}^Nc_{ij}a_{ij}{(\tilde x_i-\tilde x_j)^T\Gamma(\tilde x_i-\tilde x_j)},
\end{aligned}
\end{equation}
and
\begin{equation}\label{young1'}
\begin {aligned}
&\sum_{i=1}^N\sum_{j=1}^Nc_{ij}a_{ij}{(e_i-e_j)^T\Gamma(e_i-e_j)}\\
&\qquad\leq 2\sum_{i=1}^N\sum_{j=1}^Nc_{ij}a_{ij}{e_i^T\Gamma e_i}+2\sum_{i=1}^N\sum_{j=1}^Nc_{ij}a_{ij}{e_j^T\Gamma e_j}\\
&\qquad=4\sum_{i=1}^N\sum_{j=1}^Nc_{ij}a_{ij}{e_i^T\Gamma e_i}.
\end{aligned}
\end{equation}
Substituting \dref{xid}, \dref{young1}, and \dref{young1'} into \dref{lya1d1'} yields
\begin{equation}\label{lya1d2}
\begin {aligned}
\dot V_1
&\leq \sum_{i=1}^N{\xi_i^TPA\xi_i}-\frac{\alpha}{4}
\sum_{i=1}^N\sum_{j=1}^Na_{ij}{(\tilde x_i-\tilde x_j)^T\Gamma(\tilde x_i-\tilde x_j)}\\
&\quad+\sum_{i=1}^N\sum_{j=1}^Nc_{ij}a_{ij}e_i^T\Gamma e_i
-\sum_{i=1}^N\sum_{j=1}^N\frac{c_{ij}-\alpha}{4}\varrho_{ij}a_{ij}c_{ij}\\
&\leq \sum_{i=1}^N{\xi_i^TPA\xi_i}-\frac{\alpha}{4}
\sum_{i=1}^N\sum_{j=1}^Na_{ij}{(\tilde x_i-\tilde x_j)^T\Gamma(\tilde x_i-\tilde x_j)}\\
&\quad+\sum_{i=1}^N\sum_{j=1}^Nc_{ij}a_{ij}e_i^T\Gamma e_i
+\varsigma-\sum_{i=1}^N\sum_{j=1}^N\frac{\varrho_{ij}a_{ij}}{8}(c_{ij}-\alpha)^2,
\end{aligned}
\end{equation}
where we have used the Young's inequality to get the last inequality and $\varsigma=\sum_{i=1}^N\sum_{j=1}^N\frac{\varrho_{ij}a_{ij}}{8}\alpha^2$.

Note that
\begin{equation}\label{lya1d2'}
\begin {aligned}
&\sum_{i=1}^N\sum_{j=1}^Na_{ij}{(\tilde x_i-\tilde x_j)^T\Gamma(\tilde x_i-\tilde x_j)}\\
&\qquad=\sum_{i=1}^N\sum_{j=1}^Na_{ij}{(\xi_i-\xi_j)^T\Gamma(\xi_i-\xi_j)}\\
&\qquad\quad+\sum_{i=1}^N\sum_{j=1}^Na_{ij}{(e_i-e_j)^T\Gamma(e_i-e_j)}\\
&\qquad\quad+2\sum_{i=1}^N\sum_{j=1}^Na_{ij}{(x_i-x_j)^T\Gamma(e_i-e_j)},
\end{aligned}
\end{equation}
\begin{equation}\label{young2}
\begin {aligned}
&-\sum_{i=1}^N\sum_{j=1}^Na_{ij}{(x_i-x_j)^T\Gamma(e_i-e_j)}\\
&\qquad\leq \frac{1}{4}\sum_{i=1}^N\sum_{j=1}^Na_{ij}{(x_i-x_j)^T\Gamma(x_i-x_j)}\\
&\qquad\quad+\sum_{i=1}^N\sum_{j=1}^Na_{ij}{(e_i-e_j)^T\Gamma(e_i-e_j)},
\end{aligned}
\end{equation}
and
\begin{equation}\label{young2'}
\begin {aligned}
\sum_{i=1}^N\sum_{j=1}^Na_{ij}{(e_i-e_j)^T\Gamma(e_i-e_j)}
\leq 4\sum_{i=1}^N\sum_{j=1}^Na_{ij}e_i^T\Gamma e_i.
\end{aligned}
\end{equation}
Substituting \dref{lya1d2'}, \dref{young2}, and \dref{young2'} into \dref{lya1d2} gives
\begin{equation}\label{lya1d2''}
\begin {aligned}
\dot V_1
&\leq \frac{1}{2}\xi^T\left[I_N\otimes(PA+A^TP)-\frac{\alpha}{4}\mathcal L\otimes \Gamma\right]\xi\\
&\quad+
\frac{\alpha}{2}\sum_{i=1}^N[\sum_{j=1}^N(1+\frac{2}{\delta \alpha}\cdot\delta c_{ij})a_{ij}e_i^T\Gamma e_i\\
&\quad-\frac{1}{4}\sum_{j=1}^Na_{ij}{(\tilde x_i-\tilde x_j)^T\Gamma(\tilde x_i-\tilde x_j)}]\\
&\quad-\sum_{i=1}^N\sum_{j=1}^N\frac{\varrho_{ij}a_{ij}}{8}(c_{ij}-\alpha)^2
+\varsigma.
\end{aligned}
\end{equation}

By the definition of $\xi$, it is easy to see that $({\mathbf{1}}^T\otimes I)\xi=0$.
Because $\mathcal G$ is connected, it then follows from Lemma 1 that $\xi^T(\mathcal L\otimes \Gamma)\xi\geq \lambda_2(\mathcal L)\xi^T(I_N\otimes \Gamma)\xi$, where $\lambda_2(\mathcal L)$ is the smallest nonzero eigenvalue of $\mathcal L$.
By noting the triggering functions \dref{eve} and choosing $\alpha$ to be sufficiently large such that $\alpha\geq \max\{\frac{2}{\delta},\frac{4}{\lambda_2(\mathcal L)}\}$, it follows from \dref{lya1d2''} that
\begin{equation}\label{lya1d3}
\begin {aligned}
\dot V_1
&\leq \frac{1}{2}\xi^T\left[I_N\otimes(PA+A^TP)-\frac{\alpha}{4}\mathcal L\otimes PBB^TP\right]\xi\\
&\quad-\sum_{i=1}^N\sum_{j=1}^N\frac{\varrho_{ij}a_{ij}}{8}(c_{ij}-\alpha)^2
+\varsigma+\frac{\alpha}{2}N\mu e^{-\nu t}\\
&\leq -\frac{1}{2}\xi^T\xi
-\sum_{i=1}^N\sum_{j=1}^N\frac{\varrho_{ij}a_{ij}}{8}(c_{ij}-\alpha)^2+\varsigma+\frac{\alpha}{2}N\mu e^{-\nu t}.
\end{aligned}
\end{equation}
Substituting \dref{lya1} into \dref{lya1d3} yields
\begin{equation*}\label{V1-d4}
\begin {aligned}
\dot V_1
&\leq -\theta_1 V_1+\frac{1}{2}\theta_1\xi^T(I_N\otimes P)\xi -\frac{1}{2}\xi^T\xi\\
&\quad \!+\!\frac{1}{8}\sum_{i=1}^N\sum_{j=1}^N(\frac{\theta_1}{\kappa_{ij}}\!
-\!\varrho_{ij})a_{ij}(c_{ij}-\alpha)^2
+\varsigma+\frac{\alpha}{2}N\mu e^{-\nu t}\\
&\leq -\theta_1 V_1
+\varsigma+\frac{\alpha}{2}N\mu e^{-\nu t},
\end{aligned}
\end{equation*}
where $\theta_1=\mathrm{min}_{(v_i,v_j)\in \mathcal E}\{\varrho_{ij}\kappa_{ij},\frac{1}{\lambda_{\max}(P)}\}$.
According to the Comparison lemma \cite{ZLi2014cooperative}, we have
\begin{equation*}\label{V3-int}
\begin {aligned}
V_1(t)&\leq [V_1(0)-\frac{\varsigma}{\theta_1}]e^{-\theta_1 t}
+\frac{\varsigma}{\theta_1}
+\frac{\alpha}{2}N\psi(t),
\end{aligned}
\end{equation*}
where $\psi(t)$ is defined as
\begin{equation*}\label{psi}
\begin {aligned}
\psi(t)=\begin{cases}\mu te^{-\theta_1 t}&\text{if}~\theta_1=\nu,\\
\frac{\mu}{\theta_1-\nu}(e^{-\nu t}-e^{-\theta_1 t})&\text{if}~\theta_1\neq \nu.
\end{cases}
\end{aligned}
\end{equation*}
It is not difficult to verify that $\lim_{t\rightarrow +\infty}\psi(t)=0$.
Therefore, $V_1$ exponentially converges to a bounded set 
$\mathcal S_1\triangleq \left\{\xi,~c_{ij}~|~V_1\leq \frac{\varsigma}{\theta_1}\right\}.$
In light of the fact that $V_1\geq \frac{\lambda_{\min}(P)}{2}\|\xi\|^2+\sum_{i=1}^N\sum_{j=1}^N\frac{(c_{ij}-\alpha)^2}{8\kappa_{ij}}$,
we conclude that $\xi$ and $c_{ij}$ are all uniformly ultimately bounded.

Since $\theta_2 \triangleq \min_{\forall (v_i,v_j)\in \mathcal E}\{\varrho_{ij}\kappa_{ij}\}<\frac{1}{\lambda_{\max}(P)}$, we can rewrite \dref{lya1d3} as
\begin{equation}\label{V1-d5}
\begin {aligned}
\dot V_1
&\leq -\theta_2 V_1+\frac{1}{2}\theta_2\xi^T(I_N\otimes P)\xi -\frac{1}{2}\xi^T\xi\\
& +\frac{1}{8}\sum_{i=1}^N\sum_{j=1}^N(\frac{\theta_2}{\kappa_{ij}}-\varrho_{ij})a_{ij}(c_{ij}-\alpha)^2
+\varsigma+\frac{\alpha}{2}N\mu e^{-\nu t}\\
&\leq -\theta_2 V_1-\rho\xi^T\xi
+\varsigma+\frac{\alpha}{2}N\mu e^{-\nu t},
\end{aligned}
\end{equation}
where $\rho=\frac{1}{2}(1-\theta_2 \lambda_{\max}(P))$.
Obviously, it follows from \dref{V1-d5} that $\dot V_1\leq -\theta_2 V_1+\frac{\alpha}{2}N\mu e^{-\nu t}$ if $\xi^T\xi> \frac{\varsigma}{\rho}$.
Therefore, we can obtain that the consensus error $\xi$ exponentially converges to the following bounded set:
 \begin{equation}\label{D2}
\begin {aligned}
\mathcal S_2\triangleq \left\{\xi~|~\|\xi\|^2\leq \frac{\varsigma}{\rho}\right\}.
\end{aligned}
\end{equation}
This completes the proof.
$\hfill $$\blacksquare$
\end{proof}

Theorem 1 shows that the consensus error $\xi$ under \dref{pro1} and \dref{eve} converges to a residual set that can be arbitrary small by choosing proper constants $\varrho_{ij}$. The term $-\varrho_{ij}c_{ij}(t)$ in \dref{pro1} is inspired by the $\sigma$-modification technique in the adaptive control literature \cite{ioannou1996robust,ZLi2014cooperative}. When $-\varrho_{ij}c_{ij}(t)$ is removed from \dref{pro1}, in this case the consensus error $\xi$ will asymptotically converge to zero, as stated in the following corollary.

\begin{corollary}
Let $\varrho_{ij}=0$, $\forall (v_i,v_j)\in\mathcal E$ in \dref{pro1}.
Under the conditions as in Theorem \ref{theorem_1}, the consensus error $\xi$ asymptotically converges to zero.
\end{corollary}

One advantage of the event-based adaptive protocol \dref{pro1} including $-\varrho_{ij}c_{ij}(t)$ is that it 
is robust in presence of external disturbances or uncertainties. For instance, the robustness of \dref{pro1} with respect to perturbed agents $\dot x_i=Ax_i+Bu_i+w_i$, where $w_i$ are bounded disturbances, can be shown by following similar steps in the proof of Theorem 1 with a few modifications. The upper bound of the consensus error in this case will depend on both $\varrho_{ij}$ and the upper bounds of $w_i$. The details are skipped here due to the space limitation.

\begin{remark}
It is well known that a necessary and sufficient condition for
the existence of a $P>0$ to the ARE \dref{lyaine} is that $(A,B)$ is stabilizable.
Therefore, a sufficient condition for the existence of the adaptive
protocol \dref{pro1} and \dref{eve} satisfying Algorithm \ref{algo1} and Theorem 1 is that $(A,B)$ is stabilizable.
\end{remark}

\begin{remark}
The final consensus value $\varpi(t)$ reached by the agents can be established. Using \dref{pro1} for \dref{sys}, we obtain that 
$\dot x=(I_N\otimes A)x+(\mathcal L^c\otimes BK)\tilde{x},$
where $\mathcal L^c$ is defined as $\mathcal L_{ii}^c=\sum_{j=1,j\neq i}^N{c_{ij}a_{ij}}$ and $\mathcal L_{ij}^c=-c_{ij}a_{ij}$, $i\neq j$, and $\tilde{x}=[\tilde{x}_1^T,\cdots,\tilde{x}_N^T]^T$. Noting that $\mathcal L^c$ is a weighted symmetric Laplacian matrix of $\mathcal G$, it is not difficult to verify that $({\bf 1}^T\otimes e^{-At})x$ is an invariant quantity. Therefore, $({\bf 1}^T\otimes e^{-At})({\bf 1}\otimes \varpi(t))=({\bf 1}^T\otimes I) x_0$, from which we can derive that  $\varpi(t)=\frac{1}{N}\sum_{i=1}^N e^{At}x_i(0)$.
\end{remark}

\begin{remark}
In Theorem \ref{theorem_1} and Corollary 1, the communication graph is assumed to be fixed throughout the whole process.
Actually, the proposed adaptive consensus protocol is applicable to the case of arbitrary switching communication graphs with a positive
dwelling time, which are connected at each contiguous time interval.
In this case, communications only take place when the triggering conditions are violated or the topology switches.
It is not challenging to prove this assertion by taking $V_1$ in \dref{lya1} as a common Lyapunov function.
\end{remark}

\begin{remark}
Theorem 1 and Corollary 1 show that the agents in \dref{sys} can reach consensus under the proposed event-based adaptive protocol, consisting of the control law \dref{pro1} and the triggering function \dref{eve},  for any connected communication topology.
Contrary to the protocols in the previous works \cite{Seyboth2013event,DYang2016Decentralized,Zhu2014event,Garcia2014decentralized,TCheng2017event}, which require global information of the communication graph in terms of the nonzero eigenvalues of the corresponding Laplacian matrix, the adaptive protocol in the current paper is fully distributed and scalable, relying on none global information of the network graph and independent of the network's scale.
\end{remark}

The following theorem excludes the Zeno behavior.

\begin{theorem}\label{theorem_1+}
Under the conditions in Theorem \ref{theorem_1}, the network \dref{xid} does not exhibit the Zeno behavior and the interval between two consecutive triggering instants for any agent is strictly positive, as illustrated in \dref{tau}. 
\end{theorem}

\begin{proof}
For agent $v_i$, consider the evolution of $e_i(t)$ for $t\in [t_k^i,t_{k+1}^i)$, $t_{k+1}^i<\infty$. It follows from \dref{sys}, \dref{ei}, and \dref{pro1} that
\begin{equation*}\label{edot}
\begin {aligned}
\dot{e}_i=Ae_i-\sum_{j=1}^Nc_{ij}a_{ij}BK(\tilde x_i-\tilde x_j).
\end{aligned}
\end{equation*}
The time derivative of $\|e_i(t)\|$ for $t\in [t_k^i,t_{k+1}^i)$ can be then obtained as
\begin{equation}\label{eidotnorm}
\begin {aligned}
\frac{d\|e_i(t)\|}{dt}&=\frac{e_i^T}{\|e_i\|}\dot{e}_i\leq \|\dot e_i\|\\&\leq \|A\|\|e_i\|+\sum_{j=1}^Nc_{ij}a_{ij}\|BK\|\|\tilde x_i-\tilde x_j\|.
\end{aligned}
\end{equation}

As shown in Theorem 1, both $c_{ij}(t)$ and $\xi$ are bounded, the latter of which implies that $x_i(t)-x_j(t)$, $\forall (v_i,v_j)\in\mathcal E$, is bounded.
Without loss of generality, assume that $c_{ij}\leq \bar c$ for some positive constant $\bar c$.
Note that the interval between two consecutive triggering events is bounded. Thus, $e^{A(t-t_k^i)}$ is bounded for any $t\in[t_k^i,t_{k+1}^i)$.
Since $e^{-At}\sum_{j=1}^N{x_j}$ is an invariant quantity (see Remark 3),  we can derive from $ x_i=\xi_i+\frac{1}{N}\sum_{j=1}^N{x_j}$ that $x(t)$ is finite for any finite $t$. Therefore, we can get that 
 for any $t\in[t_k^i,t_{k+1}^i)$, $\tilde x_i-\tilde x_j=e^{A(t-t_k^i)}x_i(t_k^i)-e^{A(t-t_{k'}^j)}x_j(t_{k'}^j)$ is also bounded, where $t_{k'}^j$ denotes the latest event triggering instant of agent $v_j$. 
 Then, it follows from \dref{eidotnorm} that 
\begin{equation}\label{eidotnorm0}
\begin {aligned}
\frac{d\|e_i(t)\|}{dt}&\leq \|A\|\|e_i\|+\bar c \sigma_i,
\end{aligned}
\end{equation}
where $\sigma_i$ denotes the upper bound of $\sum_{j=1}^Na_{ij}\|BK\|\|\tilde x_i-\tilde x_j\|$ for $t$ from $t_k^i$ to $t_{k+1}^i$. Consider a non-negative function $\varphi:[0,\infty)\rightarrow \mathbf R_{\geq 0}$, satisfying
\begin{equation}\label{varphi}
\begin {aligned}
\dot \varphi=\|A\|\varphi+\bar c\sigma_i,~\varphi(0)=\|e_i(t_k^i)\|=0.
\end{aligned}
\end{equation}
Then, we have that $\|e_i(t)\|\leq \varphi(t-t_k^i)$, where $\varphi(t)$ is the analytical solution to \dref{varphi}, given by
$\varphi(t)=\frac{\bar c\sigma_i}{\|A\|}(e^{\|A\|t}-1).$

It is not difficult to see that the triggering function \dref{eve} satisfies $f_i(t)\leq0$, if we have the following condition:
\begin{equation}\label{einorm}
\begin {aligned}
\|e_i\|^2\leq \frac{\mu e^{-\nu t}}{d_i\|K\|^2(1+\delta \bar c)}.
\end{aligned}
\end{equation}
In light of \dref{einorm}, it is easy to see that the interval between two triggering instants $t_k^i$ and $t_{k+1}^i$  for agent $v_i$ can be lower bounded by the time for $\varphi^2(t-t_k^i)$ evolving from $0$ to the right-hand side of \dref{einorm}. Therefore, 
a lower bound $\tau_k^i$ of $t_{k+1}^i-t_k^i$ can be obtained by solving the following inequality:
\begin{equation*}\label{varphi_sol2}
\begin {aligned}
\frac{\bar c^2\sigma_i^2}{\|A\|^2}(e^{\|A\|\tau_k^i}-1)^2&\geq\frac{\mu e^{-\nu (t_k^i+\tau_k^i)}}{d_i\|K\|^2(1+\delta \bar c)}.
\end{aligned}
\end{equation*}
Then, we get that
\begin{equation}\label{tau}
\begin {aligned}
t_{k+1}^i-t_k^i&\geq\tau_k^i\\&\geq\frac{1}{\|A\|}\ln \left(1+\frac{\|A\|}{\bar c\sigma_i\|K\|}\sqrt{\frac{\mu e^{-\nu (t_k^i+\tau_k^i)}}{d_i(1+\delta \bar c)}}\right).
\end{aligned}
\end{equation}
Note that $\tau_k^i$ always exists and is strictly positive for any finite time. 
Since the right-hand side of the second inequality in \dref{tau} approaches zero only when $t\rightarrow\infty$, we get that $t_k^i\rightarrow\infty$ with $k\rightarrow\infty$ and $t\rightarrow\infty$. Therefore, Zeno behavior is excluded for all the agents for any finite time.
$\hfill $$\blacksquare$
\end{proof}

\begin{remark}
It is worth mentioning that the proof of Theorem \ref{theorem_1+} is partly inspired by \cite{Garcia2014decentralized,TCheng2017event,suninpress}.
Note that the lower bound $\tau_k^i$ for the inter-event intervals in \dref{tau} is generally conservative, since it is derived by using only  the exponential decay term  in the triggering function \dref{eve} and ignoring the effect of the state-dependent term.  Nevertheless, an advantage of the state-dependent term in \dref{eve} is that it can significantly reduce the number of event triggering, which can also be verified by numerical simulations.  
Besides,  the lower bound $\tau_k^i$ in \dref{tau} depends on the specific time instants, not uniform with respect to $t$, and approaches zero when $t\rightarrow\infty$. If we replace $\mu e^{-\nu t}$ in \dref{eve} by a small positive constant $\mu$, then the lower bound $\tau_k^i$ for the inter-event intervals satisfies 
$\tau_k^i\geq\frac{1}{\|A\|}\ln (1+\frac{\|A\|}{\bar c\sigma_i\|K\|}\sqrt{\frac{\mu}{d_i(1+\delta \bar c)}})$, which is always positive for any time. The cost is that in this case asymptotical convergence of the consensus error cannot be guaranteed and only practical consensus can be expected. 
\end{remark}

\subsection{Observer-Based Adaptive Event-Triggered Protocols}\label{s4}

In this subsection, we consider the case where only the local output information is available.

We propose for each agent the following adaptive event-based control law:
\begin{equation}\label{pro2}
\begin {aligned}
\dot \chi_i&=A\chi_i+Bu_i+F(C\chi_i-y_i),\\
\dot c_{ij}&=
\kappa_{ij}a_{ij}[-\varrho_{ij}c_{ij}+(\tilde \chi_i-\tilde \chi_j)^T\Gamma(\tilde \chi_i-\tilde \chi_j)],\\
u_i&=K\sum_{j=1}^N{c_{ij}a_{ij}(\tilde \chi_i-\tilde \chi_j)},~i=1,\cdots,N,
\end{aligned}
\end{equation}
where $\chi_i$ is the estimate of the state $x_i$ of agent $v_i$,
$\tilde \chi_i(t)=e^{A(t-t_k^i)}\chi_i(t_k^i)$, $F$, $K$, and $\Gamma$ are feedback gain matrices to be designed,
and the rest of the variables are defined as in \dref{pro1}.

We define the measurement error $e_i(t)\triangleq \tilde \chi_i-\chi_i$, $i=1,\cdots,N$.
The triggering function for each agent $v_i$ is given by
\begin{equation}\label{eve2}
\begin {aligned}
f_i(t)&=\sum_{j=1}^N(1+\delta c_{ij})a_{ij}e_i^T\Gamma e_i\\
&\quad-\frac{1}{4}\sum_{j=1}^N{a_{ij}(\tilde \chi_i-\tilde \chi_j)^T\Gamma(\tilde \chi_i-\tilde \chi_j)}
-\mu e^{-\nu t},
\end{aligned}
\end{equation}
where $\delta$, $\mu$, and $\nu$ are positive constants.

Let $x=[x_1^T,\cdots,x_N^T]^T$, $\chi=[\chi_1^T,\cdots,\chi_N^T]^T$, and $\tilde \chi=[\tilde \chi_1^T,\cdots,\tilde \chi_N^T]^T$.
Denote the consensus error by $\zeta=[\zeta_1^T,\cdots,\zeta_N^T]^T=(M\otimes I_n)x$ and $\eta=[\eta_1^T,\cdots,\eta_N^T]^T=(M\otimes I_n)\chi$, where $M=I_N-\frac{1}{N}{\mathbf{1}\mathbf{1}}^T$, $\zeta_i=x_i-\frac{1}{N}\sum_{j=1}^Nx_j$, and $\eta_i=\chi_i-\frac{1}{N}\sum_{j=1}^N\chi_j$, $i=1,\cdots,N$.
Then, we can get from \dref{sys} and \dref{pro2} that
\begin{equation}\label{xi2}
\begin {aligned}
\dot \zeta&=(I_N\otimes A)\zeta+(\mathcal L^c\otimes BK)\tilde \chi,\\
\dot \eta&=(I_N\otimes A)\eta+(\mathcal L^c\otimes BK)\tilde \chi+(I_N\otimes FC)(\eta-\zeta),
\end{aligned}
\end{equation}
where $\mathcal L_c$ is defined as in Remark 3.

\begin{algorithm}\label{algo2}
Assuming that Assumption \ref{ABC} holds, the observer-based event-triggered adaptive consensus protocol consisting of \dref{pro2} and \dref{eve2} can be designed according to the following steps.\\
1) Choose the feedback matrix $F$ such that $A+FC$ is Hurwitz (One such $F$ can be chosen as $F=-\tilde PC^T$,
where $\tilde P>0$ is the solution to the ARE: $\tilde PA^T+A\tilde P-\tilde PC^TC\tilde P+I=0$).\\
2)-4) The same as steps 1) to 3) in Algorithm \ref{algo1}.
\end{algorithm}

\begin{theorem}\label{theorem_2}
Suppose that Assumption 1 holds.
Both the consensus error $\xi$ and the coupling gains $c_{ij}$ in \dref{pro2} are uniformly ultimately bounded under the event-triggered adaptive protocol constructed by Algorithm \ref{algo2}.
Furthermore, the Zeno behavior can be excluded.
\end{theorem}

\begin{proof}
Let $\varepsilon=\eta-\zeta$.
Then, \dref{xi2} can be rewritten in terms of $\varepsilon$ and $\eta$ as
\begin{equation}\label{vare}
\begin {aligned}
\dot \varepsilon&=[I_N\otimes (A+FC)]\varepsilon,\\
\dot \eta&=(I_N\otimes A)\eta+(\mathcal L^c\otimes BK)\tilde \chi+(I_N\otimes FC)\varepsilon.
\end{aligned}
\end{equation}
Evidently, consensus is achieved if $\varepsilon$ and $\eta$ asymptotically converge to zero.

Since $A+FC$ is Hurwitz, it is well known that there exists a $\tilde Q>0$ such that $\tilde Q(A+FC)+(A+FC)^T\tilde Q=\omega I$, where $\omega$ is a positive constant.
Let
\begin{equation}\label{lya21}
\begin {aligned}
V_{21}=\varepsilon^T(I_N\otimes \tilde Q)\varepsilon.
\end{aligned}
\end{equation}
The time derivative of $V_{21}$ along the trajectory of \dref{vare} is given by
\begin{equation}\label{lya21d1}
\begin {aligned}
\dot V_{21}
&=\varepsilon^T\{I_N\otimes [\tilde Q(A+FC)+(A+FC)^T\tilde Q]\}\varepsilon
=-\omega \varepsilon^T\varepsilon.
\end{aligned}
\end{equation}
It is easy to see from \dref{lya21d1} that $\dot V_{21}<0$, implying that $\varepsilon(t)\rightarrow 0$ as $t\rightarrow \infty$.

Let
\begin{equation}\label{lya22}
\begin {aligned}
V_{22}=\frac{1}{2}\eta^T(I_N\otimes P)\eta+\sum_{i=1}^N\sum_{j=1,j\neq i}^N{\frac{(c_{ij}-\alpha)^2}{8\kappa_{ij}}},
\end{aligned}
\end{equation}
where $\alpha$ is a positive constant to be determined later.
The time derivative of $V_{22}$ along the trajectory of \dref{vare} is given by
\begin{equation}\label{lya22d1}
\begin {aligned}
\dot V_{22}
&=\eta^T(I_N\otimes PA)\eta+\eta^T(\mathcal L^c\otimes PBK)\tilde \chi\\
&\quad+\eta^T(I_N\otimes PFC)\varepsilon+\sum_{i=1}^N\sum_{j=1,j\neq i}^N{\frac{c_{ij}-\alpha}{4\kappa_{ij}}\dot c_{ij}}.
\end{aligned}
\end{equation}
Using the Young's inequality gives
\begin{equation}\label{fangsuo1}
\begin {aligned}
\eta^T(I_N\otimes PFC)\varepsilon
&\leq \frac{1}{4}\eta^T(I_N\otimes Q)\eta+\frac{\|PFC\|^2}{\lambda_{\min}(Q)}\varepsilon^T \varepsilon.
\end{aligned}
\end{equation}

Consider the following Lyapunov function candidate
\begin{equation}\label{lya2}
\begin {aligned}
V_2=\frac{\|PFC\|^2}{\omega\lambda_{\min}(Q)}V_{21}+V_{22}.
\end{aligned}
\end{equation}
Evidently, $V_2$ is positive definite.
By using \dref{lya21d1}, \dref{lya22d1}, and \dref{fangsuo1}, we can obtain the time derivative of $V_2$ as
\begin{equation*}\label{lya2d1}
\begin {aligned}
\dot V_2&\leq \eta^T(I_N\otimes PA)\eta+\eta^T(\mathcal L^c\otimes PBK)\tilde \chi+\frac{1}{4}\eta^T(I_N\otimes Q)\eta\\
&\quad+\sum_{i=1}^N\sum_{j=1,j\neq i}^N{\frac{c_{ij}-\alpha}{4\kappa_{ij}}\dot c_{ij}}\\
&\leq \frac{1}{2}\eta^T[I_N\otimes (PA+A^TP+\frac{1}{2}Q)]\eta+\eta^T(\mathcal L^c\otimes PBK)\tilde \chi\\
&\quad+\sum_{i=1}^N\sum_{j=1,j\neq i}^N{\frac{c_{ij}-\alpha}{4\kappa_{ij}}\dot c_{ij}}.
\end{aligned}
\end{equation*}
Following similar steps as in the proof of Theorem 1, by using the triggering functions \dref{eve2} and choosing $\alpha$ sufficiently large such that $\alpha\geq \max\{\frac{2}{\delta},\frac{4}{\lambda_2(\mathcal L)}\}$, we can obtain that
\begin{equation*}\label{lya2d2}
\begin {aligned}
\dot V_2&\leq \frac{1}{2}\eta^T[I_N\otimes (PA+A^TP+\frac{1}{2}Q)-\frac{\alpha}{4}\mathcal L\otimes \Gamma]\eta
+\varsigma+\frac{\alpha}{2}N\mu e^{-\nu t}\\
&\leq-\frac{1}{4}\lambda_{\min}(Q)\eta^T\eta+\varsigma+\frac{\alpha}{2}N\mu e^{-\nu t},
\end{aligned}
\end{equation*}
where $\varsigma=\sum_{i=1}^N\sum_{j=1}^N\frac{\varrho_{ij}a_{ij}}{8}\alpha^2$.

Similarly as in the proof of Theorem \ref{theorem_1}, we can show that $\eta(t)$ is uniformly ultimately bounded.
Because $\varepsilon(t)\rightarrow 0$ and $\eta(t)\rightarrow 0$ as $t\rightarrow \infty$, it follows that the consensus error $\zeta(t)$ is also uniformly ultimately bounded.

The Zeno behavior can be excluded similarly as in Theorem \ref{theorem_1+}. The details are omitted here for brevity.
$\hfill $$\blacksquare$
\end{proof}

\begin{remark}
Theorem \ref{theorem_2} shows that, under the observer-based adaptive protocols \dref{pro2} and the triggering functions \dref{eve2}, the consensus error $\zeta(t)$ is uniformly ultimately
bounded.
A sufficient condition for the existence of \dref{pro2} and \dref{eve2} satisfying Algorithm \ref{algo2} and Theorem \ref{theorem_2} is that the pair $(A,B,C)$ is stabilizable and detectable.
\end{remark}

\section{Fully Distributed Event-based Protocols For Leader-Follower Consensus}\label{s-leader}

In this section, we extend to consider the event-triggered consensus problem in the presence of one leader.
Without loss of generality, assume that the agent indexed by $v_1$ is the leader whose control input $u_1$ is supposed to be zero.
The communication graph $\mathcal G$ among the agents is assumed to satisfy the following assumption.

\begin{assumption}\label{assumption_2}
The pair $(A,B)$ in \dref{sys} is stabilizable.
The subgraph associated with the followers is undirected and the graph $\mathcal G$ contains a directed spanning tree with the leader as the root.
\end{assumption}

Because the leader has no neighbors, the Laplacian matrix $\mathcal L$ can be partitioned as
$
\mathcal L=
\left[\begin{smallmatrix}
     0   &0_{1\times (N-1)}\\
     \mathcal L_2 & \mathcal L_1
\end{smallmatrix}\right],
$
where $\mathcal L_1\in \mathbf{R}^{(N-1)\times (N-1)}$ is symmetric and $\mathcal L_2\in \mathbf{R}^{(N-1)\times 1}$.
In light of Lemma \ref{lemma_eigenvalue}, $\mathcal L_1$ is positive definite.

In the following, we consider only the case where the relative state information is available.
For each follower, we propose the following adaptive event-based control law:
\begin{equation}\label{pro3}
\begin {aligned}
u_i&=K\sum_{j=1}^N{c_{ij}a_{ij}(\tilde x_i-\tilde x_j)},~i=2,\cdots,N,\\
\dot c_{ij}&=
\kappa_{ij}a_{ij}[-\varrho_{ij}c_{ij}+(\tilde x_i-\tilde x_j)^T\Gamma(\tilde x_i-\tilde x_j)],
\end{aligned}
\end{equation}
where $\tilde x_1(t)=e^{At}x_1(0)$, $\tilde x_i(t)=e^{A(t-t_k^i)}x_i(t_k^i)$, $i=2,\cdots,N$, $c_{ij}(t)$ denotes the time-varying coupling weight for the edge $(v_i,v_j)$ with $c_{ij}(0)=c_{ji}(0)$ for $i=2,\cdots,N$, $j=1,\cdots,N$, $\kappa_{ij}=\kappa_{ji}$ are positive constants, and $K\in \mathbf R^{p\times n}$ and $\Gamma \in \mathbf R^{n\times n}$ are the feedback gain matrices.

The triggering function for each follower $v_i$ is designed as
\begin{equation}\label{eve3}
\begin {aligned}
f_i(t)&=\frac{1}{2}(1+\delta c_{i1})a_{i1}e_i^T\Gamma e_i
+\sum_{j=2}^N(1+\delta c_{ij})a_{ij}e_i^T\Gamma e_i\\
&\quad-\frac{1}{4}\sum_{j=1}^Na_{ij}(\tilde x_i-\tilde x_j)^T\Gamma (\tilde x_i-\tilde x_j)-\mu e^{-\nu t}\\
&\quad-\frac{1}{2}a_{i1}(\tilde x_i-\tilde x_1)^T\Gamma (\tilde x_i-\tilde x_1),~i=2,\cdots,N,
\end{aligned}
\end{equation}
where $e_i=\tilde x_i-x_i$, $i=2,\cdots,N$, and $\delta$, $\mu$, and $\nu$ are positive constants.

Define the consensus error $z_i=x_i-x_1$, $i=2,\cdots,N$.
Let $\tilde z_i=\tilde x_i-\tilde x_1$.
By noting that $\dot x_1=Ax_1$, it is not difficult to see that $\tilde x_1(t)\equiv x_1(t)$.
Then, we can get from \dref{sys} and \dref{pro3} that
\begin{equation}\label{vid}
\begin {aligned}
\dot z_i&=Az_i+BK\sum_{j=2}^N{c_{ij}a_{ij}(\tilde z_i-\tilde z_j)}+BKc_{i1}a_{i1}\tilde z_i,\\
\dot c_{i1}&=\kappa_{i1}a_{i1}[-\varrho_{i1}c_{i1}+\tilde z_i^T\Gamma \tilde z_i],\\
\dot c_{ij}&=
\kappa_{ij}a_{ij}[-\varrho_{ij}c_{ij}+(\tilde z_i-\tilde z_j)^T\Gamma(\tilde z_i-\tilde z_j)].
\end{aligned}
\end{equation}

\begin{theorem}
Suppose that Assumption 2 holds.
Choose $K=-B^TP$ and $\Gamma=PBB^TP$, where $P>0$ is defined as in \dref{lyaine}.
Both the consensus error $\xi$ and the coupling gains $c_{ij}$ in \dref{pro3} are uniformly ultimately bounded under the event-triggered adaptive protocol \dref{pro3} and \dref{eve3}.
Furthermore, the closed-loop system does not exhibit the Zeno behavior.
\end{theorem}

\begin{proof}
Consider the Lyapunov function candidate
\begin{equation}\label{lya3}
\begin {aligned}
V_3=\frac{1}{2}\sum_{i=2}^N{z_i^TPz_i}+\sum_{i=2}^N\sum_{j=2,j\neq i}^N{\frac{(c_{ij}-\beta)^2}{8\kappa_{ij}}}+\sum_{i=2}^N{\frac{(c_{i1}-\beta)^2}{4\kappa_{i1}}},
\end{aligned}
\end{equation}
where $\beta$ is a positive constant to be determined later.
Evidently, $V_3$ is positive definite.
The time derivative of $V_3$ along the trajectory of \dref{vid} is given by
\begin{equation}\label{lya3d1}
\begin {aligned}
\dot V_3
&=\sum_{i=2}^N{z_i^TPAz_i}-\sum_{i=2}^N{c_{i1}a_{i1}z_i^T\Gamma \tilde z_i}
-\sum_{i=2}^N{\sum_{j=2}^N{c_{ij}a_{ij}z_i^T\Gamma (\tilde z_i-\tilde z_j)}}\\
&\quad+\sum_{i=2}^N{\frac{c_{i1}-\beta}{2}a_{i1}[-\varrho_{i1}c_{i1}+\tilde z_i^T\Gamma \tilde z_i]}\\
&\quad+\sum_{i=2}^N\sum_{j=2,j\neq i}^N{\frac{c_{ij}-\beta}{4}a_{ij}[-\varrho_{ij}c_{ij}+(\tilde z_i-\tilde z_j)^T\Gamma(\tilde z_i-\tilde z_j)]}.
\end{aligned}
\end{equation}

Note that
\begin{equation}\label{lya3d2}
\begin {aligned}
-\sum_{i=2}^N{c_{i1}a_{i1}z_i^T\Gamma \tilde z_i}
&=-\sum_{i=2}^N{c_{i1}a_{i1}\tilde z_i^T\Gamma \tilde z_i}+\sum_{i=2}^N{c_{i1}a_{i1}e_i^T\Gamma \tilde z_i},
\end{aligned}
\end{equation}
\begin{equation}\label{lya3d3}
\begin {aligned}
\sum_{i=2}^N{c_{i1}a_{i1}e_i^T\Gamma \tilde z_i}
&\leq \frac{1}{2}\sum_{i=2}^N{c_{i1}a_{i1}\tilde z_i^T\Gamma \tilde z_i}+\frac{1}{2}\sum_{i=2}^N{c_{i1}a_{i1}e_i^T\Gamma e_i},
\end{aligned}
\end{equation}
\begin{equation}\label{lya3d4}
\begin {aligned}
&-\sum_{i=2}^N{\sum_{j=2}^N{c_{ij}a_{ij}z_i^T\Gamma (\tilde z_i-\tilde z_j)}}\\
&\qquad=-\frac{1}{2}\sum_{i=2}^N{\sum_{j=2}^N{c_{ij}a_{ij}(\tilde z_i-\tilde z_j)^T\Gamma (\tilde z_i-\tilde z_j)}}\\
&\qquad\quad+\frac{1}{2}\sum_{i=2}^N{\sum_{j=2}^N{c_{ij}a_{ij}(e_i-e_j)^T\Gamma (\tilde z_i-\tilde z_j)}},
\end{aligned}
\end{equation}
and
\begin{equation}\label{lya3d5}
\begin {aligned}
&\sum_{i=2}^N{\sum_{j=2}^N{c_{ij}a_{ij}(e_i-e_j)^T\Gamma (\tilde z_i-\tilde z_j)}}\\
&\qquad\leq \frac{1}{2}\sum_{i=2}^N{\sum_{j=2}^N{c_{ij}a_{ij}(\tilde z_i-\tilde z_j)^T\Gamma (\tilde z_i-\tilde z_j)}}\\
&\qquad\quad+\frac{1}{2}\sum_{i=2}^N{\sum_{j=2}^N{c_{ij}a_{ij}(e_i-e_j)^T\Gamma (e_i-e_j)}}.
\end{aligned}
\end{equation}
Substituting \dref{vid}, \dref{lya3d2}, \dref{lya3d3}, \dref{lya3d4}, and \dref{lya3d5} into \dref{lya3d1} yields
\begin{equation}\label{lya3d6}
\begin {aligned}
\dot V_3
&\leq\sum_{i=2}^N{z_i^TPAz_i}-\frac{\beta}{2}\sum_{i=2}^N{a_{i1}\tilde z_i^T\Gamma \tilde z_i}+\frac{1}{2}\sum_{i=2}^N{c_{i1}a_{i1}e_i^T\Gamma e_i}\\
&\quad-\frac{\beta}{4}\sum_{i=2}^N{\sum_{j=2}^N{a_{ij}(\tilde z_i-\tilde z_j)^T\Gamma (\tilde z_i-\tilde z_j)}}\\
&\quad+\frac{1}{4}\sum_{i=2}^N{\sum_{j=2}^N{c_{ij}a_{ij}(e_i-e_j)^T\Gamma (e_i-e_j)}}
+\varsigma',
\end{aligned}
\end{equation}
where
$\varsigma'=\sum_{i=2}^N\frac{\varrho_{i1}a_{i1}}{4}\beta^2
+\sum_{i=2}^N\sum_{j=2}^N\frac{\varrho_{ij}a_{ij}}{8}\beta^2$.
Let $z=[z_2^T,\cdots,z_N^T]^T$.
Similarly as did in deriving \dref{lya1d2''}, we can obtain from \dref{lya3d6} that
\begin{equation*}\label{lya3d7}
\begin {aligned}
\dot V_3
&\leq \frac{1}{2}z^T\left[I_{N-1}\otimes (PA+A^TP)-\frac{\beta}{4}\mathcal L_1\otimes \Gamma\right]z\\
&\quad+\frac{\beta}{2}\sum_{i=2}^N\{\frac{1}{2}(1+\frac{2}{\delta\beta}\cdot \delta c_{i1})a_{i1}e_i^T\Gamma e_i\\
&\quad
+\sum_{j=2}^N(1+\frac{2}{\delta\beta}\cdot \delta c_{ij})a_{ij}e_i^T\Gamma e_i-\frac{1}{2}a_{i1}(\tilde x_i-\tilde x_1)^T\Gamma (\tilde x_i-\tilde x_1)\\
&\quad-\frac{1}{4}\sum_{j=2}^Na_{ij}(\tilde x_i-\tilde x_j)^T\Gamma (\tilde x_i-\tilde x_j)\}+\varsigma'.
\end{aligned}
\end{equation*}
By noting the triggering functions \dref{eve3} and choosing $\beta$ to be sufficiently large such that $\beta\geq \max\{\frac{2}{\delta},\frac{4}{\lambda_{\min}(\mathcal L_1)}\}$, we have
\begin{equation*}\label{lya3d8}
\begin {aligned}
\dot V_3
&\leq -\frac{1}{2}\lambda_{\min}(Q)z^Tz+\varsigma'+\frac{\beta}{2}(N-1)\mu e^{-\nu t}.
\end{aligned}
\end{equation*}

The rest of the proof can be shown similarly as in Theorems \ref{theorem_1} and \ref{theorem_1+}. The details are omitted
here for brevity.
$\hfill $$\blacksquare$
\end{proof}

\section{Simulation Examples}\label{s-simulation}
In this section,
we illustrate the effectiveness of the theoretical results by doing numerical simulations.
For illustration, consider the linear multi-agent  systems described
by \dref{sys}, with $A=\left[\begin{smallmatrix}0&1&0\\0&0&1\\0&0&0\end{smallmatrix}\right]$ and $B=\left[\begin{smallmatrix}0\\0\\1\end{smallmatrix}\right].$
All initial values of the agents are randomly chosen.
The communication graph among the agents is depicted in Fig. \ref{fig},
which evidently satisfies Assumption \ref{ABC}.
\begin{figure}[!htb]
\centering
\includegraphics[width=0.18\textheight,height=0.2\textwidth]{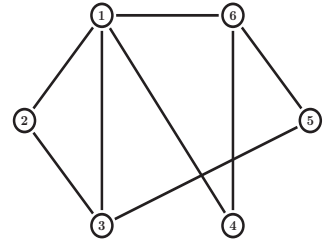}
\caption{The undirected communication graph among agents.}
\label{fig}
\end{figure}
\begin{figure}
  \centering
  \subfigure[Fig. 2a]{
    \label{fig:subfig:a} 
    \includegraphics[width=0.15\textheight,height=0.15\textwidth]{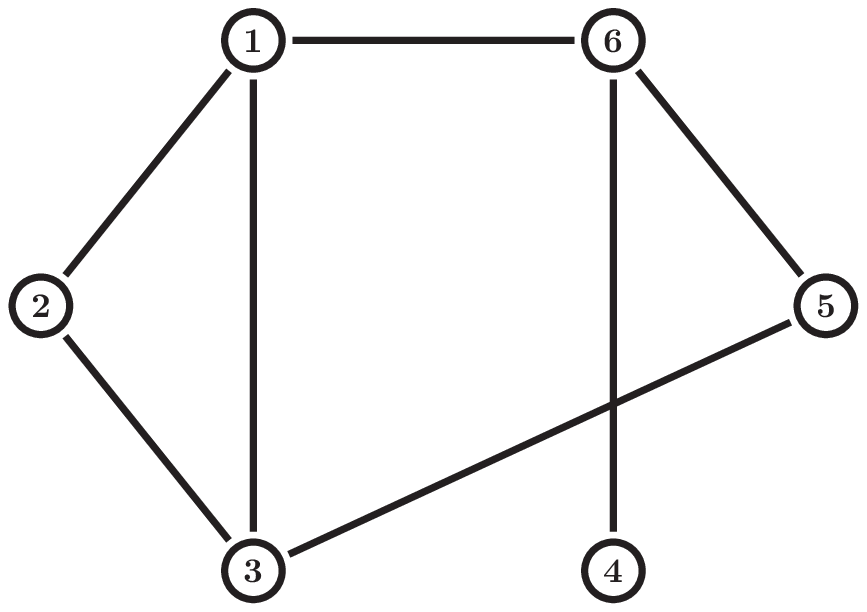}}
\quad
  \subfigure[Fig. 2b]{
    \label{fig:subfig:b} 
    \includegraphics[width=0.145\textheight,height=0.145\textwidth]{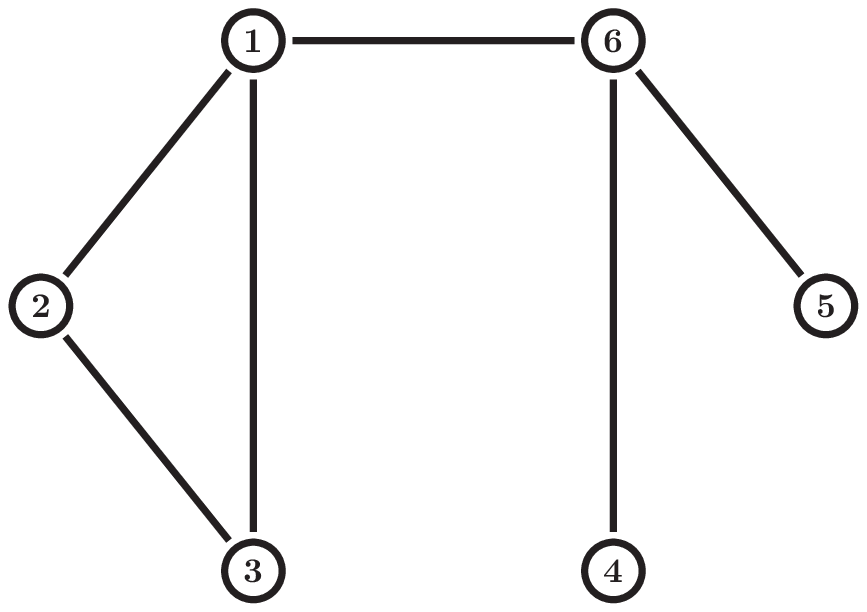}}
  \caption{Two other undirected communication graphs.}
  \label{fig:subfig} 
\end{figure}

Solving the ARE \dref{lyaine} gives a solution:
$P=\left[\begin{smallmatrix}2.4142&2.4142&1.0000\\2.4142&4.8284&2.4142\\1.0000&2.4142&2.4142\end{smallmatrix}\right].$
Thus, the feedback matrices in the event-based control law \dref{pro1} can be obtained as
$K=\left[\begin{smallmatrix}-1.0000&-2.4142&-2.4142\end{smallmatrix}\right],$
$\Gamma=\left[\begin{smallmatrix}1.0000&2.4142&2.4142\\2.4142&5.8284&5.8284\\2.4142&5.8284&5.8284\end{smallmatrix}\right].$
Other parameters of \dref{pro1} and \dref{eve} are chosen as $\delta=1$, $\mu=2$, $\nu=0.5$, $\varrho_{ij}=0$, and $\kappa_{ij}=0.2$, $\forall (v_i,v_j)\in \mathcal E$.
The consensus errors $x_i-x_1$, $i=2,\cdots,N$, are depicted in Fig. \ref{fig-consensuserror}, from which we can observe that consensus is indeed achieved.
The adaptive coupling weights $c_{ij}(t)$ in \dref{pro1} are shown in Fig. \ref{c_ij}, which implies that $c_{ij}(t)$ converge to finite steady-state values.
The triggering instants of each agent are presented in Fig. \ref{fig-triggering}, which shows that the Zeno behavior is excluded.
To study how the network topology affects triggering time instants, we do simulations from 0s to 20s, respectively, under three different topologies, i.e., these in Fig. 1, Fig. 2a, and Fig. 2b.

\begin{figure}[!htb]
\centering
\includegraphics[width=0.3\textheight,height=0.25\textwidth]{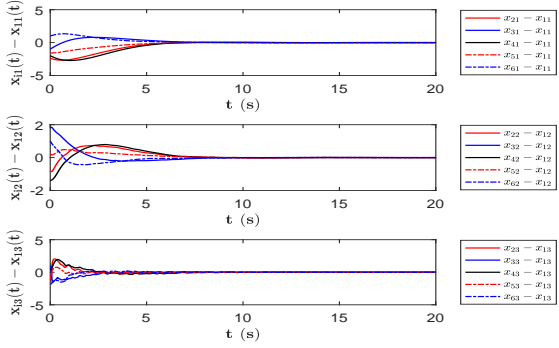}
\caption{The consensus errors $x_i-x_1$, $i=2,\cdots,6$.}
\label{fig-consensuserror}
\end{figure}
\begin{figure}[!htb]
\centering
\includegraphics[width=0.3\textheight,height=0.2\textwidth]{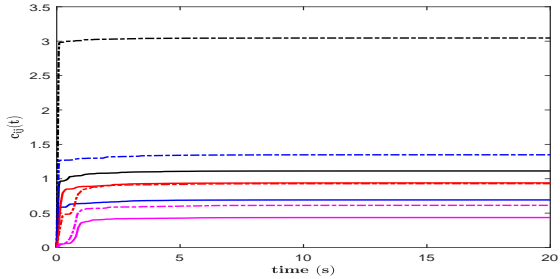}
\caption{The adaptive coupling weights $c_{ij}(t)$ in \dref{pro1}.}
\label{c_ij}
\end{figure}
\begin{figure}[!htb]
\centering
\includegraphics[width=0.3\textheight,height=0.2\textwidth]{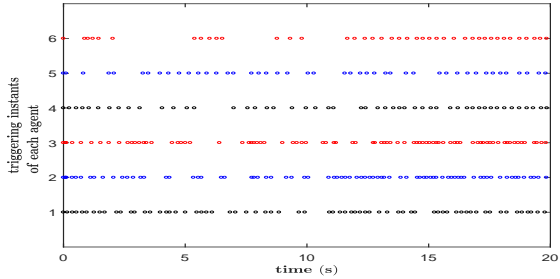}
\caption{Triggering instants of each agent.}
\label{fig-triggering}
\end{figure}

\section{Conclusion}\label{s_con}
In this paper,
we have designed distributed adaptive event-based protocols to solve both the leaderless and leader-follower consensus problem for general linear multi-agent networks.
Compared to the previous related works,
our main contribution is that we have proposed for the first time in the literature fully distributed and scalable consensus protocols, which do not rely on any global information of the network graph and are independent of the network's scale.
Our event-triggered protocols do not need continuous communications among
neighboring agents and do not exhibit the Zeno behavior.
Extending the results to general directed graphs is an interesting work for future study.

\end{document}